\documentclass[submission,copyright,creativecommons]{eptcs}
\providecommand{\event}{FOMC 2013} 
\usepackage{breakurl}            

\usepackage[ruled,vlined,linesnumbered]{algorithm2e}
\usepackage{amssymb,amsmath,amsfonts,booktabs}
\usepackage{verbatim}

\newtheorem{theorem}{Theorem}[section]
\newtheorem{lemma}[theorem]{Lemma}

\newtheorem{corollary}[theorem]{Corollary}
\newtheorem{definition}{Definition}
\newtheorem{observation}{Observation}

\newcommand{\sq}{\hbox{\rlap{$\sqcap$}$\sqcup$}}
\newcommand{\qed}{\hspace*{\fill}\sq}
\newenvironment{proof}{\noindent {\bf Proof.}\ }{\qed\par\vskip 4mm\par}
\newenvironment{proofs}{\noindent {\bf Proof (sketch).}\ }{\qed\par\vskip 4mm\par}

\newcommand{\cE}{{\cal E}}
\newcommand{\cO}{{\cal O}}
\newcommand{\mN}{\mathbb{N}}
\newcommand{\bQ}{\mathbf{Q}}
\newcommand{\cU}{{\cal U}}
\newcommand{\bB}{\mathbf{B}}

\newcommand{\dist}{{\rm dist}}
\newcommand{\queue}{{\rm queue}}
\newcommand{\cancel}{{\rm cancel}}

\title{Distributed Queuing in Dynamic Networks}
\author{Gokarna Sharma
\institute{School of Electrical Engineering and Computer Science\\
Louisiana State University\\
Baton Rouge, LA 70803, USA}
\email{gokarna@csc.lsu.edu}
\and
Costas Busch
\institute{School of Electrical Engineering and Computer Science\\
Louisiana State University\\
Baton Rouge, LA 70803, USA}
\email{busch@csc.lsu.edu}
}
\def\titlerunning{Distributed Queuing in Dynamic Networks}
\def\authorrunning{G. Sharma \& C. Busch}
\begin{document}
\maketitle

\begin{abstract}
We consider the problem of forming a distributed queue in the adversarial dynamic network model of Kuhn, Lynch, and Oshman (STOC 2010) in which the network topology changes from round to round but the network stays connected. This is a synchronous model in which network nodes are assumed to be fixed, the communication links for each round are chosen by an adversary, and nodes do not know who their neighbors are for the current round before they broadcast their messages. 
Queue requests may arrive over rounds at arbitrary nodes and the goal is to eventually enqueue them in a distributed queue. 
We present two algorithms that give a total distributed ordering of queue requests in this model. We measure the performance of our algorithms through {\em round complexity}, which is the total number of rounds needed to solve the distributed queuing problem. We show that in 1-interval connected graphs, where the communication links change arbitrarily between every round, it is possible to solve the distributed queueing problem in $O(nk)$ rounds using $O(\log n)$ size messages, where $n$ is the number of nodes in the network and $k\leq n$ is the number of queue requests. Further, we show that for more stable graphs, e.g. $T$-interval connected graphs where the communication links change in every $T$ rounds, the distributed queuing problem can be solved in $O\left(n+ \frac{nk}{\min\{\alpha,T\}}\right)$ rounds using the same $O(\log n)$ size messages, where $\alpha > 0$ is the concurrency level parameter that captures the minimum number of active \queue\ requests in the system in any round.
These results hold in any arbitrary (sequential, one-shot concurrent, or dynamic) arrival of $k$ \queue\ requests in the system.
Moreover, our algorithms ensure correctness in the sense that each queue request is eventually enqueued in the distributed queue after it is issued and each queue request is enqueued exactly once. We also provide an impossibility result for this distributed queuing problem in this model. To the best of our knowledge, these are the first solutions to the distributed queuing problem in adversarial dynamic networks.
\end{abstract}

\section{Introduction}
\label{section:introduction}

Many distributed systems rely on some concept of mobile objects.
A mobile object lives on only one node of the network at a time and it moves from
one node to another in response to explicit requests by other nodes. 
A tracking
mechanism, commonly known as a distributed directory, allows nodes to keep track of mobile objects by providing
the ability to locate the objects and also the ability to ensure consistency of the objects in concurrent situations \cite{Demmer1998}. These directories are interesting in the sense that they provide the controlled way of sending the mobile object to the requesting nodes without flooding the object information to the whole network.

This mobile object tracking problem has been extensively studied in the literature for
various coordination problems that arise in a distributed setting.
For example, authors in \cite{Naimi1996,Raymond1989} studied this problem in the context of distributed mutual exclusion. The node which has the token (or the shared object) can enter the critical section in their problem setting.
Later, Demmer and Herlihy \cite{Demmer1998} studied this problem in the context of distributed directories.
Awerbuch and Peleg  \cite{Awerbuch1990} studied this problem in the context of tracking a mobile user in sensor networks. Recently, these papers \cite{Herlihy2007,Sharma2012,Sharma2012b,Attiya2010,Zhang2009} studied this problem for distributed transactional memories.
In these applications, the concept of path reversal - {\em when a node receives a message,
it flips its edge to point to the node from which the request was received} - is used.
Path reversal approaches are evolved from the trail of forwarding pointers approach of \cite{Li1989} studied for memory coherence in virtual shared memory systems.

The very common feature of the aforementioned solutions \cite{Demmer1998,Naimi1996,Raymond1989,Awerbuch1990,Herlihy2007,Sharma2012,Sharma2012b,Attiya2010,Zhang2009} is that they essentially form some short of a {\em distributed queue} by which processes (i.e. nodes or vertices) that issued operations for a shared object can be organized in a total order and each processor receives the identity of its predecessor operation in that total order \cite{Demmer1998,Herlihy2006,Herlihy2001}. A distributed queue approach is appealing because it ensures that no single node becomes a synchronization bottleneck \cite{Demmer1998,Naimi1996,Raymond1989}.
However, these previous solutions
assumed the static network such that a pre-selected spanning tree \cite{Demmer1998,Naimi1996,Raymond1989,Zhang2009} or a hierarchical directory \cite{Awerbuch1990,Herlihy2007,Sharma2012,Sharma2012b,Attiya2010} can be embedded on top of the graph. It is yet to know whether it is possible to come up with efficient solutions to the queuing problem in dynamic graphs. This is because
when the network topology changes frequently, there might be a significant overhead on adapting the commonly used structures (tree or hierarchy) accordingly to
cope up with the changes. 

Therefore,
we initiate the study of the distributed queuing problem in situations where the underlying network graph changes frequently such that a static structure can not be efficiently maintained.
To model frequent changes,
we consider the adversarial dynamic network model first studied by Kuhn, Lunch, and Oshman \cite{Kuhn2010}. This is a synchronous model in which time is divided into rounds, and in each round, the communication network is a graph chosen by an adversary over a vertex set. The vertex set is assumed to be fixed throughout the execution. The communication graph is also assumed to be connected but it can change completely from one round to the next, i.e. the network topology changes from round to round. Nodes communicate by broadcasting messages to their immediate neighbors.
The adversary is very strong in the sense that nodes do not know who their neighbors are for the current round before they broadcast their messages.
This model is appealing in the sense that it captures widely-used mobile and wireless networks where communication can be unpredictable (see \cite{ODell2005,Kuhn2011b} for details). Our main objective in this present work
is to understand the complexity of the distributed queuing problem in
this adversarial dynamic network model.

\paragraph{Contributions:} Assume that there are $n$ nodes in the network and $k\leq n$ nodes issue a queue request each which 
must be ordered in such a way that each requesting node receives the identity of its predecessor node in a total distributed order.
We derive an impossibility result showing that this distributed queuing problem can not be solved without {\queue} request replication in adversarial dynamic networks.
We then give two simple algorithms for this problem, one for frequently changing graphs and the other for more stable graphs, assuming that the adversary satisfies {\em $T$-interval connectivity}: there must exist a connected spanning subgraph that stably persists throughout $T$ consecutive rounds. To be more clear, $T$-interval connected graphs are those graphs in which for any consecutive interval of $T$ rounds, the maximal common subgraph of the graphs in these rounds is connected. The communication is 
limited to $O(\log n)$ bits per message. 


We measure the performance of our algorithms through {\em round complexity}, which is the total number of rounds needed to solve the queuing problem. Our goal is to minimize the total number of rounds needed in solving the queuing problem.
We show that in 1-interval connected graphs, where the communication links change arbitrarily between rounds, our algorithm needs $O(nk)$ rounds to solve the queueing problem. 
Further, we show that in more stable graphs, e.g. $T$-interval connected graphs where the communication links change in every known $T>1$ rounds, our algorithm needs
$O\left(n+ \frac{nk}{\min\{\alpha,T\}}\right)$ rounds to solve the queuing problem, 
where $\alpha>0$ is a concurrency level parameter that captures the minimum number of {\em active} (initiated but not yet enqueued)  \queue\ requests in the system in any round.
These bounds hold in all three cases of sequential, (one-shot) concurrent, and dynamic execution of $k$ {\queue} requests. A sequential execution consists of non-overlapping sequence of {\queue} operations, 
whereas a set of {\queue} requests are initiated simultaneously in a concurrent execution. 
For dynamic executions, we consider a window of time such that an arbitrary set of bounded $k$ {\queue} requests are assumed to be initiated at arbitrary moments of time within that window.
Therefore, sequential and concurrent executions are the special cases of dynamic executions.
Let us denote by {\em cycle} the window of $O(n)$ consecutive rounds and by $\beta_i$ the number of  active queue requests in the beginning of cycle $i$. The value of $\beta_i$ may be different from cycle to cycle depending on the execution, however  $1 \leq \beta_i \leq k$ holds for every cycle $i$ in any execution. Therefore, $\alpha$ that appears in the bound $O\left(n+ \frac{nk}{\min\{\alpha,T\}}\right)$ is essentially the smallest value of $\beta_i$ in any cycle $i$, 
i.e., $\alpha := \min\{\beta_1,\beta_2,\ldots\}$. 
%
%
%
%
This bound is interesting in the sense that it shows that the performance speed up can only be obtained in $T$-interval connected graphs for the distributed queuing problem when $\alpha \approx T$ throughout the execution. 

Our results also extend to dynamic executions with continuous arrival of \queue\ requests from nodes over time (i.e., $k\rightarrow \infty$) where we show that, if $\beta_i$ are the active \queue\ requests in the beginning of any cycle $i$, then our algorithms guarantee that they will be enqueued within next $O(n\beta_i)$ rounds in 1-interval connected graphs, and within next $O\left(n+ \frac{n\beta_i}{\min\{\beta_i,T\}}\right)$ rounds in $T$-interval connected graphs.
Moreover, our algorithms ensure correctness in any execution (see Section \ref{section:limitation} for details) in the sense that each queue request is eventually enqueued in the distributed queue after it is issued and each queue request is enqueued exactly once.
To our best knowledge, these are the first solutions to the distributed queuing problem in adversarial dynamic networks.


Our bounds suggest that the queuing problem needs as much as the number of rounds needed to solve the {\em counting} problem\footnote{In the counting problem, assuming that nodes do not know $n$ in the initial state, every node in the dynamic graph comprising $n$ nodes should know $n$ after some rounds of message exchange \cite{Kuhn2010}.} and the {\em $k$-token dissemination} problem\footnote{In the $k$-token dissemination problem, there are $k$ unique tokens, usually in $k$ different nodes of the network, and the goal is to transmit these tokens to all the nodes in the network \cite{Kuhn2010}.}
in dynamic networks, in the worst-case. It is shown that $O(n^2)$ rounds are sufficient \cite{Kuhn2010} and $\Omega(n^2/\log n)$ rounds are necessary \cite{Dutta2013} to solve the counting and all-to-all token dissemination problems.
The complexity arises in adversarial dynamic networks due to the fact that the communication graph changes in every round. Therefore, even in the case of distributed queuing, a queue request may need to visit all the rest $n-2$ nodes before finding the predecessor node, which takes $n-1$ rounds as the communication in each round is controlled by the worst-case adversary. Someone may say that the distributed queuing problem can be solved by first solving the $k$-token dissemination problem and then making the predecessor a node with ID that is immediately smaller than that of a queue request issuing node
for every node that issued the queue request. However, this approach only solves the queuing problem in the case of an one-shot concurrent execution (and does not solve the problem in sequential and dynamic executions).

\paragraph{Related Work:}
The distributed queuing problem has been studied extensively in the literature assuming a static network \cite{Demmer1998,Naimi1996,Raymond1989,Awerbuch1990,Herlihy2007,Sharma2012,Sharma2012b,Attiya2010,Zhang2009}.
To solve this problem, either the pre-selected spanning tree as used in \cite{Demmer1998,Naimi1996,Raymond1989,Attiya2010,Zhang2009} or the  hierarchical structure as used in \cite{Awerbuch1990,Herlihy2007,Sharma2012,Sharma2012b} is constructed on top of the static network.
These ideas were based on some well-known spanning tree and clustering
techniques (e.g., minimum spanning tree \cite{Cormen2009}, sparse covers \cite{Awerbuch1990}, maximal independent sets \cite{Luby1985}) which organize the nodes in the network in some useful way to facilitate efficient coordination. These papers \cite{Herlihy2001,Herlihy2007,Attiya2010,Sharma2012,Sharma2012b}
studied the distributed queuing problem in the concurrent execution setting, and these papers \cite{Sharma2013,Herlihy2006} considered dynamic executions. Moreover, the self-stabilizing version of the distributed queuing problem was studied by Tirthapura and Herlihy \cite{Tirthapura2006}. This self-stabilizing version is also not inherently dynamic as the eventual stabilization of the network is assumed, i.e., the network stabilizes and stops changing after a finite time.
These approaches, e.g. \cite{Demmer1998,Naimi1996,Raymond1989,Awerbuch1990,Herlihy2007,Sharma2012}, 
used latency as the cost metric, i.e., the cost 
is
measured through the total latency, which is the sum
of the latencies of 
individual queuing requests. 
However, in dynamic networks, the problem is to figure out how many rounds of message broadcasts are required to solve the distributed queuing problem. 

The adversarial dynamic network model was proposed in the seminal paper of Kuhn, Lynch, and Oshman \cite{Kuhn2010}. The authors studied the complexity of counting and token dissemination problems. Subsequently,
there have been a significant interest in solving many distributed coordination problems in this model as it makes very few assumptions about the behavior of the network, such that the properties of the highly dynamic large scale mobile and sensor networks can be captured. 
Kuhn et al.~\cite{Kuhn2011} studied the
problem of coordinated consensus in this model. Recently, these papers \cite{Haeupler2012,Dutta2013} improved  and extended some of the  
results  presented in \cite{Kuhn2010}. Moreover, Haeupler and Karger,  in \cite{Haeupler2011}, studied how to use network coding to expedite the information dissemination in this model. We direct readers to \cite{Kuhn2011b} for the state-of-the-art up to the end of 2010.

\paragraph{Outline of Paper:} The rest of the paper is organized as follows. In Section \ref{section:preliminaries}, we formally present the adversarial dynamic network model and define the distributed queuing problem.
We give a very simple impossibility result in Section \ref{section:impossibility}. We then present and analyze a queuing algorithm for frequently changing graphs in Section \ref{section:frequent}. 
We do the same for more stable graphs in Section \ref{section:stable}. We then discuss an inherent limitation 
in Section \ref{section:limitation} and conclude with a short discussion in Section \ref{section:discussion}. 

\section{Preliminaries}
\label{section:preliminaries}
\subsection{Dynamic Network Model}
We formally present the dynamic network model, originally introduced by Kuhn, Lynch, and Oshman \cite{Kuhn2010}.  
This model works on a synchronous round based execution.
A dynamic network is represented as a connected graph $G=(V,E)$, where $|V|=n$. We assume that $n$ is known to the nodes of $G$. If $n$ is not known, an existing counting algorithm, e.g. \cite{Kuhn2010}, can be used to find $n$ spending $O(n^2)$ rounds.  This is not a much overhead as counting is needed only once, whereas queuing is an ongoing service.
Each vertex of $G$ models a node, and each edge a two-way reliable communication link. Moreover, each node has a unique identifier (UID) drawn from a namespace $\cU$. These identifers have $O(\log n)$ bits, so that they fit in a message.
Each node can send messages directly to its neighbors and indirectly to non-neighbors along a path. Each edge has same weight and sending a message from one node to its neighbor node needs a single round. It is assumed that every message is eventually delivered (i.e. no message loss occurs).

This model assumes that nodes share a common global clock that starts at 0 and advances in unit steps. The communication is done in synchronous rounds as follows \cite{Kuhn2010}: The round $r$ starts as soon as round $r-1$ finishes. The time between time $r-1$ and time $r$ is assumed to be the round $r$ and the following execution happens in each round $r$. First, each node generates a single message to broadcast based on its local state at time $r-1$. The adversary then provides connected communication graph (i.e., a set of edges) for round $r$. Each node then delivers its message to it's neighbors following the edges given by the adversary. The assumption of connected communication graph is each round is the only constraint on the adversary. After messages are delivered to the neighbors, each node processes the messages it received, and transits to a new state (its state at time $r$). Then, the next round begins. The communication is assume to be limited to $O(\log n)$ bits per message.

The adversary is actually a {\em strong adaptive adversary} in the sense that it can decide the network $G(r)$ of round $r$ based on the complete history of the network up to round $r-1$ as well as on the messages the nodes will send in round $r$. 
Formally, the adversary's behavior in a given execution is captured by dynamic graph $G=(V,E)$, where $V$ is a static set of nodes and $E: \mN \rightarrow \{\{u,v\}|u,v\in V\}$ is a function that maps a round number $r\in \mN$ to a set of undirected edges $E(r)$. $\dist(u,v)$ is used to denote the minimum hop distance between nodes $u,v\in G$ in the dynamic subgraph $G(r)$ at round $r$. $G$ satisfies the following property. 

\begin{definition}[\cite{Kuhn2010}]
\label{definition:t-interval}
A dynamic graph $G=(V,E)$ is said to be {\em $T$-interval connected} for any $T\geq 1$ if for all $r\in \mN$, the static graph $G_{r,T}:=\left(V,\bigcap_{i=r}^{r+T-1} E(r)\right)$ is connected. The graph is said to be {\em $\infty$-interval connected} if there is a connected static subgraph $G'= (V,E')$ such that for all $r\in \mN$, $E'\subseteq E(r)$.
\end{definition}

A dynamic graph $G=(V,E)$ in this model  induces a {\em casual order}, denoted $(u,r) \rightsquigarrow_G (v,r')$, which means that node $u$'s state in round $r$ influences node $v$'s state in round $r'$. 
The casual order is a transitive and reflexive closure of the order $(u,r) \rightarrow_{G} (v,r+1)$, which holds if and only if either $u=v$ or $\{u,v\} \in E(r+1)$. Therefore at round $r$, node $u$ has direct information about the states of node $v$ at round $r'$ such that $(v,r') \rightsquigarrow_{G} (u,r)$. 
The following lemma shows that the number of nodes that have influenced a node $u$ grows by at least one in every round, which is a very important property for this model.
\begin{lemma}[\cite{Kuhn2010}]
\label{lemma:unecut-more}
For any node $u\in V$ and round $r\geq 0$, $|\{v\in V : (u,0) \rightsquigarrow (v,r)\}|\geq \min\{r+1,n\}$ and $|\{u\in V : (v,0) \rightsquigarrow (u,r)\}|\geq \min\{r+1,n\}$.
\end{lemma}

\subsection{Distributed Queuing Problem}
We denote a distributed queue $\bQ =(h,g,\ldots,t)$ by an UID set of $|\bQ|$ nodes, where the first node $h\in \bQ$ is the head of the queue and the last node $t\in \bQ$ is the tail of the queue.
Initially, there is only one node in $\bQ$ which acts as both the head and the tail of the queue; the tail changes when other requests change the tail of the queue by becoming the successor. For example, $g$ is the {\em successor} node of $h$ and $h$ is the {\em predecessor} node of $g$ in $\bQ$.  $\bQ$ is not explicitly known to all the nodes in the system and is maintained implicitly by the nodes. A predecessor node stores only the UID of its successor node in the queue. Therefore, by visiting the successor nodes of all the nodes in $\bQ$ starting from its head provides the total distributed queue order.

An instance of the distributed queuing problem consists of a
set $\cE=\{q_1,q_2,\ldots,q_k\}$ of $k$ {\queue} requests which want to join $\bQ$.
An algorithm solves the queuing problem if for all instances $\cE$, when the algorithm is executed in any dynamic graph $G=(V,E)$, all {\queue} requests are eventually organized one after another providing a total distributed order. Each queue request $q_i$ has a source node $s_i$, which is the node that issued this request, and a destination node $t_i$, which is its predecessor node in the queue. In the distributed queuing problem, the source node of the predecessor
request $q_i$ in the total order is the destination node for the successor request
$q_{i+1}$, i.e., the destination node for each request is not known beforehand and
the distributed queuing algorithm should find out the destination node online while in execution. The purpose of any queuing algorithm is to provide the total order.

We denote a {\queue} request $q\in \cE$ by
the tuple $q = (r, u)$, where $r \geq 0$ is the time when the {\queue} request is initiated and $u$ is the node that initiates
it (i.e., the requesting node). Therefore, we denote by
$\cE= \{q_1 = (r_1, v_1), q_2 = (r_2, v_2), \ldots, q_k = (r_k, v_k)\}$ the arbitrary set of $k$ dynamic
{\queue} requests, where the requests $r_i \in \cE$ are indexed according to their initiation
times, i.e. $i < j \Longrightarrow r_i \leq  r_j$.
We also consider sequential and concurrent (one-shot) execution of these {\queue} requests.
In a sequential case, the requests in $\cE$ have initiation times such that they provide a
non-overlapping sequence of $k$ {\queue} operations, i.e., a next request will be issued
only after the current request finishes. In one-shot concurrent case, the requests in $\cE$ have same initiation times
such that all $k$ {\queue} requests come to the system at the same time. 

\section{An Impossibility Result}
\label{section:impossibility}

We prove a very simple impossibility result for the distributed queuing problem showing the power of the adversary in the dynamic graph model. 
We mean by queue request replication that when a node receives a queue request from some other node, it stores a copy in it before forwarding that queue request to its neighbors.
This theorem shows that {\queue} request replication in network nodes is necessary to solve the distributed queueing problem in the adversarial dynamic network model.

\begin{theorem}
\label{theorem:impossibility}
The distributed queuing problem is impossible to solve in 1-interval connected graphs against a strong adversary without {\queue} request replication.
\end{theorem}
\begin{proofs}
We prove this theorem similar to the impossibility proof for token dissemination given in \cite{Kuhn2010}. Consider a distributed queuing problem. Assume that, initially, there exist a head node in $\bQ$, say at node $v$ (the $head$ node). This node is also the tail of $\bQ$. The node $v$ has a local variable $succ_{v}$ which is initialized to $\bot$ (i.e.,  $succ_{head}=\bot$) to imply that there is no successor of the $head$ node in $\bQ$.
Assume also that each node $w \in G$ has a local Boolean variable $queue_w$ to represent that it has a {\queue} request, denoted by $queue(w)$. $queue_w$ is initially zero, and if $queue_w=1$ for some node $w$ then $w$ is said to ``join the queue''.
Lets consider the case where some node $w \in G, w\neq v$, wants to join $\bQ$, i.e., $queue_w=1$. To join $\bQ$, node $w$ needs to sends its queue request message $queue(w)$ to one of its neighbors.
In every round exactly one node in the network has the $queue(w)$ message, and it can either keep the $queue(w)$ message or pass the $queue(w)$ message to one of its neighbors. The goal is for a  predecessor node (in this proof the node $v$) to eventually have the $queue(w)$ message in some round. This problem is impossible to solve in 1-interval connected graphs. This is because as the adversary we considered has the knowledge of which node $x$ has the $queue(w)$ message, it can provide that node $x$ with only one edge $\{w,x\}$ such that $x$ is not the predecessor node for $queue(w)$. Node $x$ then has no choice except to communicate with node $w$. After $x$ receives the $queue(w)$ message, the adversary can turn around and remove all of $x$'s edges except $\{x,w\}$, so that $x$ has no choice except to pass the $queue(w)$ message back to $w$, which is the node that issued $queue(w)$. In this way the  adversary can prevent the $queue(w)$ message from ever visiting any node except $w,x$ for the \queue\ request issued by $w$.
\end{proofs}

\section{Queuing in Frequently Changing Graphs}
\label{section:frequent}

We present and formally
analyze a simple algorithm (see Algorithm \ref{algorithm:queue1interval}) which solves the distributed queuing problem in 1-interval connected graphs. Recall that the network topology changes in every round in 1-interval connected graphs. 
This algorithm is a simple extension to the token dissemination algorithm of \cite{Kuhn2010}; recall that the algorithm of \cite{Kuhn2010} solves the queuing problem only in (one-shot) concurrent situations.
This algorithm is suitable for all sequential, concurrent (one-shot), and dynamic execution of {\queue} requests (see Section \ref{section:introduction}).
Algorithm \ref{algorithm:queue1interval} is round based and runs for $k$ cycles. The value of $k$ does not need to be known to the algorithm; we discuss later how to get around to this problem. There are two phases in every cycle: the {\em search phase} and the {\em cancelation phase}. The search phase runs for $n$ rounds and after that the cancelation phase runs for the same $n$ rounds.
Therefore, each cycle is of $2n$ rounds in this algorithm. Algorithm \ref{algorithm:queue1interval} can solve the queuing problem without the cancelation phase, however in that case messages are queued in $\bQ$ in the order starting from the smallest UID message to the largest UID message.  

The intuition behind the algorithm is as follows.
In each round $r$ of the search phase, all nodes in the network propagate the smallest {\queue} request
they have heard about that has not yet joined the queue $\bQ$. 
The smallest \queue\ message request is selected with respect to the lexicographical ordering on first the initiation round and then on the UID of the requesting node of the \queue\ requests.
Initially, each node that initiated the {\queue} request broadcasts the {\queue} request to its neighbors. Moreover, in each round of the phase nodes remember the smallest {\queue} request 
they have sent or received so far in the execution, and broadcast that value in the next round of the phase. At the end of the search phase, each node in the network checks its local successor variable to determine whether a {\queue} request 
that was received during the search phase can actually join $\bQ$.

Similar to the \queue\ message broadcasting in the search phase, a special kind of message called {\cancel} message that is initiated at the predecessor node of the enqueued request at the end of the search phase, is broadcasted to the all the nodes in the network in the cancelation phase to remove the {\em pending} (i.e., waiting to join $\bQ$) {\queue} request from the network nodes for the {\queue} request that has joined $\bQ$ at the end of the search phase. Note that Algorithm \ref{algorithm:queue1interval} guarantees that at the end of every search phase one {\queue} request joins the queue; we give formal proof in Section \ref{subsection:analysis1interval}.
This {\cancel} message broadcasting is used in Algorithm \ref{algorithm:queue1interval} to ensure that every \queue\ request will be enqueued in $\bQ$ and no \queue\ request will be enqueued in $\bQ$ more than once. At the end of the cancelation phase, every node removes the matching {\queue} request, if any, from the list of {\queue} requests that are waiting at that node during execution to join $\bQ$.

\begin{algorithm}[!t]
{\footnotesize
$R_u(r) \leftarrow \emptyset$; \tcp*[f]{\queue\ requests at node $u$ at the beginning of round $r$} \\
$C_u(r) \leftarrow \emptyset$; \tcp*[f]{\cancel\ requests at node $u$ at the beginning of round $r$}\\
\BlankLine
{\bf For} $\ell=0,\ldots, k-1$ {\bf do}\\
{\Indp
{\bf Search phase:}\\
{\Indp
{\bf For} $r=0,\ldots, n-1$ {\bf do}\\

{\Indp
$q_{\min} \leftarrow$ a \queue\ message in  $R_u(r)$ that is smallest w.r.t. lexicographical ordering on the initiation round and the identifier of the issuing node, respectively;\\
{\bf broadcast} $q_{\min}$ 
to neighbors;\\
{\bf receive} \queue\ messages 
 from $s\geq 1$ neighbors;\\
$R_u(r) \leftarrow R_u(r) \bigcup \{q_1, \cdots, q_s\}$;\\
}
{\bf If} $succ_u == \bot$ {\bf then}\\
{\Indp
$t \leftarrow$ UID of the first received \queue\ message in $R_u(r)$;\\
$succ_u \leftarrow t$; \tcp*[f]{node $t$ becomes the successor of $u$}\\
{\bf generate} {\sl cancel} message $m=\langle {\sl cancel}, t\rangle$; \\ 
$C_u(r) \leftarrow C_u(r) \bigcup \{m\}$;\\
}
\BlankLine
}
{\bf Cancelation phase:}\\
{\Indp
{\bf For} $r=0,\ldots, n-1$ {\bf do}\\
{\Indp
$m \leftarrow$ the smallest UID {\sl cancel} message in $C_u(r)$;\tcp*[f]{in fact, $C_u(r)$ is a singleton set}\\
{\bf broadcast} $m$ to neighbors;\\
{\bf receive} \cancel\ messages from $s\geq 1$ neighbors;\\
$C_u(r) \leftarrow C_u(r) \bigcup \{m_1,\cdots, m_s\}$;\\
}
{\bf If} UID of the smallest {\sl cancel} message in $C_u(r)$ is equal to $UID_u$ {\bf then} 
$succ_u \leftarrow \bot$;\\
$R_u(r)\leftarrow$ $R_u(r)\backslash C_u(r)$ w.r.t. UIDs;\\ 
$C_u(r) \leftarrow \emptyset$;\\
%
}}}
\caption{A queuing algorithm run by node $u$}
\label{algorithm:queue1interval}
\end{algorithm}

We present some necessary notations used in Algorithm \ref{algorithm:queue1interval}.
We assume that, initially, there is a node in $G$ that is the head of the queue $\bQ$, denoted by $head$.
Moreover, there are two kind of requests in the system: {\queue} requests and {\cancel} requests.
We denote a {\queue} request $q$ from a node $u\in G$ by a message $m$ which is a triple $\langle {\sl queue}, r_u, UID_u\rangle$, where $r_u$ is the round in which the request $q$ was initiated and  $UID_u \in \mN$ is the unique identifier of the node $u$ that issued $q$. Moreover, we denote a {\sl cancel} request by a message $m$ which is a double  $\langle {\sl cancel}, UID_u\rangle$, where $UID_u \in \mN$ is the identifier of the node $u$ the {\queue} request from which joined $\bQ$ in some node $v$ such that $v$ issued the {\sl cancel} request to remove the pending \queue\ request $\langle {\sl queue}, r_u, UID_u\rangle$ from all nodes in $G$ except $u$ and $v$. 
Note that a corresponding {\cancel} request for a {\queue} request is always initiated by the predecessor node of that {\queue} request.

Every node $x$ in $G$ has a local variable $succ_x$ to denote the successor of the node $x$ in $\bQ$. This variable plays very important role in forming the distributed total order of the {\queue} requests. $succ_x$ variable implicitly stores the total distributed order, i.e., visiting the nodes specified by the $succ_x$ variable in the order starting from the $head$ node up to the tail node provides the distributed queuing order. The local variable 
$succ_x$ for any node $x$ takes one of the three values at any time, that is,
$succ_x \in \{UID_y, \bot, \infty\}$, where $UID_y$ is the UID of a node $y\in G$ such that $UID_x \neq UID_y$. Initially, every node $u$ in the system has $succ_u =\infty$, except the head node of the queue which has $succ_{head} =\bot$.  The value $succ_u =\infty$ for $u$ becomes $succ_u =\bot$ when $u$ becomes the successor in $\bQ$. When a {\queue} request from a node $z\in G$ finds a node $w$ with $succ_w=\bot$ (w is the tail node of $\bQ$), it changes the value of $succ_w$ from $\bot$ to the $UID_z$, the UID of $z$ to become the new tail of $\bQ$. 

We denote by $R_u(r)$ the set of {\queue} requests node $u$ has received by the beginning of round $r$. Node $u$ may or may not have the input, which we denote by $I(u)$.
Node $u$ has the input if $u$ issued the {\queue} request, otherwise it has no input. Our algorithm satisfies that: (a) for all $u\in V$ and round $r\geq 0$, the message sent by $u$ at round $r$ is a member of $R_u(r) \cup I(u) \cup \{\bot\}$, where $\bot$ denotes the empty message, and (b) node $u$ can not halt in round $r$ unless all the \queue\ requests in $\cE$ are served, i.e. enqueued in $\bQ$. Note that our algorithm do not combine or alter \queue\ messages, it only stores and forwards them.
Similarly, we denote by $C_u(r)$ the set of {\cancel} requests node $u$ has received by the beginning of round $r$. Similar to the definition of $R_u(r)$, node $u$ may or may not have a {\cancel} message as input which can be defined accordingly. 

We are now ready to describe how algorithm works.
Recall that in every cycle, the search and the cancelation phase run one after another for $n$ rounds each.
In each round of the search phase (Lines 5--9 of Algorithm \ref{algorithm:queue1interval}, the smallest {\queue} request among \queue\ requests in $R_u(r)$ is chosen to broadcast by each node $u\in V$.
The smallest \queue\ message (or request) is selected with respect to the lexicographical ordering on first the initiation round and then on the UID of the requesting node of the \queue\ requests in $R_u(r)$.
After that each node updates $R_u(r)$ by receiving the {\queue} messages send by its neighbors in that round. At the end of the search phase, each node $u\in V$ checks whether the local variable $succ_u$ is $\bot$. If $succ_u==\bot $ for some node $u$, $u$ selects the UID, say $t$, of the {\queue} message that was received by $u$ first among the available {\queue} messages in $R_u(r)$  and assign that UID to its local variable $succ_u$.
In other words, the node whose {\queue} request  reached to $v$ first becomes the successor of node $v$ (Line 11, 12 of Algorithm \ref{algorithm:queue1interval}). After that a {\cancel} message $m$ is generated at $u$ to remove the {\queue} message from $t$ (that just joined $\bQ$) that might have been replicated from the other nodes of the graph $G$ (Lines 13, 14 of Algorithm \ref{algorithm:queue1interval}). This ensures that the same \queue\ request from $t$ will not be enqueued in $\bQ$ later.

In each round of the cancelation phase (Lines 16--20 of Algorithm \ref{algorithm:queue1interval}),
each node $u$ chooses the smallest UID {\cancel} message $m$ from $C_u(r)$ and broadcast $m$ to its neighbors similar to \queue\ requests in the search phase.
After that it receives the {\cancel} messages sent to it by its neighbors and updates $C_u(r)$ accordingly. Note that as only one {\queue} request can be enqueued in $\bQ$ in the search phase, there is always only one {\cancel} message in each cancelation phase. Therefore, $C_u(r)$ is a singleton set. At the end of the cancelation phase, if the UID of node $u\in V$ matches the UID of the smallest \cancel\ message in $C_u(r)$, the node $u$ changes the value of its local variable $succ_u$ from $\infty$ to $\bot$ (Lines 21 of Algorithm \ref{algorithm:queue1interval}). At this point, one \queue\ request from some node in $G$ is served by Algorithm \ref{algorithm:queue1interval}. $R_u(r)$ for each node $u\in V$ is then updated by removing the \queue\ message from $R_u(r)$ the  UID of which matches with the UID of the \cancel\ message in $C_u(r)$, and $C_u(r)$ is made empty before transiting to the next cycle (Lines 22, 23 of Algorithm \ref{algorithm:queue1interval}).

We now describe how to get around to the problem of knowing $k$. We remove this assumption 
by allowing the node $x\in V$ which has $succ_x==\bot$ to broadcast an algorithm termination message to its neighbors if it does not receive any {\queue} message 
for up to $2n$ rounds. Node $x$ can maintain a local variable that is dedicated to perform this operation. The termination message from $x$ reaches all the nodes in the network in at most $n$ rounds; after that every node can terminate the execution.

\subsection{Analysis}
\label{subsection:analysis1interval}

\paragraph{Progress and Correctness:}

We first establish progress guarantees of Algorithm \ref{algorithm:queue1interval}. Let $\bB_i$ be the set of active \queue\ requests in the beginning of any cycle $i\geq 1$ and let
$\beta_i$ be the number of active \queue\ requests in $\bB_i$ (i.e., $\beta_i=|\bB_i|$); each cycle is of exactly $2n$ rounds in our algorithm. Note that the size of $\bB_i$ may be different from one cycle to the other cycle. Therefore, $\beta_i$ captures essentially the concurrency level of the \queue\ request execution in the algorithm in any arbitrary moment of time.
We prove progress guarantees of Algorithm \ref{algorithm:queue1interval} in dynamic executions for the continuous arrival of \queue\ requests initiated by graph nodes over time (i.e., $k$ is not bounded in this setting so that $k\rightarrow \infty$). In particular, we prove the following lemma.

\begin{lemma}
\label{lemma:onequeueeveryOn}
If there are $\beta_i$ active \queue\ requests in the beginning of any cycle $i$ in a dynamic execution, Algorithm \ref{algorithm:queue1interval} guarantees that they will be enqueued in $\bQ$ in next at most $O(n\beta_i)$ rounds.
\end{lemma}
\begin{proof}
Recall that a \queue\ request can join $\bQ$ in Algorithm \ref{algorithm:queue1interval} at the end of the search phase. 
Moreover, when searching for the node with $succ=\bot$, the \queue\ request is stored in every node it visits until this \queue\ message at those nodes is later canceled by a corresponding {\cancel} message.
We have that there is a node $head$ with $succ_{head}=\bot$ in the beginning of the execution (which is also the tail) and the tail node $tail$ in the beginning of cycle $i$  where one of the future \queue\ requests need to reach to join $\bQ$.
Therefore, we show the following for the active \queue\ requests set $\bB_i$ in the beginning of every cycle $i$:  the smallest \queue\ request (with respect to the lexicographical ordering on first the initiation round and then on the UID of the \queue\ request issuing node) in $\bB_i$, say $q_{\min}$, among $\beta_i$ requests in $\bB_i$ reaches the node $u$ with $succ_u=\bot$ within $n$ rounds from the beginning of the cycle $i$. 
This is the case because, according to Algorithm \ref{algorithm:queue1interval}, when two or more \queue\ requests reach at some intermediate node $y$ such that $succ_y=\infty$, the smallest \queue\ request among them is broadcasted to the neighboring nodes of $y$  (Lines 6, 7 of Algorithm \ref{algorithm:queue1interval}).  The node $y$ continues broadcasting the smallest {\queue} message 
among the \queue\ requests it currently holds
until $y$ receives the corresponding \cancel\ message for that {\queue} request or the other \queue\ request that is smaller than the previous one is reached to $y$ in the previous round. 
Therefore, in a given round, consider a cut between the nodes that already received the smallest \queue\ request and those that have not. From the properties of 1-interval connected graphs, there is always an edge in that cut such that when the smallest \queue\ request is broadcasted on that edge some new node receives it (Lemma \ref{lemma:unecut-more}). Since the node that initiated the \queue\ request already knows the \queue\ message and there are $n$ nodes in the graph $G$, after $n-1$ rounds all nodes have the smallest \queue\ request message.

In $n$ rounds after the beginning of cycle $i$, a \queue\ request issued by node $z$ can be enqueued in $\bQ$ by assigning $succ_{tail} \leftarrow z$, which indicates that $z$ became the successor of $tail$ in $\bQ$ in cycle $i$. The predecessor node of $z$, i.e.  $tail$, now issues a $\cancel$ message with the UID of $z$ 
and broadcasts it to its neighbors in the cancelation phase for $n$ rounds. Similar to the searching phase, 
the \cancel\ message reaches to the node that issued the \queue\ request within $n$ rounds. This can be again shown by considering the cut between the nodes that already received the \cancel\ message and those that have not (Lemma \ref{lemma:unecut-more}). When node $z$ finds the \cancel\ message with UID equals the node UID, it changes its $succ_z$ variable from $\infty$ to $\bot$ at the end of the search phase, so that other \queue\ requests can join $\bQ$ later.
Therefore, the smallest {\queue} request is finished execution by the Algorithm \ref{algorithm:queue1interval} in exactly $2n$ rounds after the beginning of cycle $i$. Now in cycle $i+1$ some other \queue\ request from $\bB_i$ becomes the smallest \queue\ request. As \queue\ requests that are initiated during cycle $i$ have initiation times greater than all the requests in $\bB_i$, they can not overtake \queue\ requests in $\bB_i$ to join $\bQ$.  That is, any request that is generated in cycle $i+1$ are ordered in the queue after the requests in $\bB_i$. Therefore, at end of cycle $i+1$, the second smallest request from $\bB_i$ joins $\bQ$. Applying this argument repetitively for $\beta_i$ requests in $\bB_i$,
all the \queue\ requests in $\bB_i$ join queue in next $\beta_i$ cycles starting from cycle $i$. Therefore, we need total $2n * \beta_i =\cO(n\beta_i)$ rounds after the beginning of the cycle $i$ to enqueue all requests in $\bB_i$.
%
\end{proof}

It is clear from Lemma \ref{lemma:onequeueeveryOn} that from the round some {\queue} request joined $\bQ$ until the round the node that issued that {\queue} request received the corresponding {\cancel} message and changes the value of its successor variable $succ$ from $\infty$ to $\bot$, $\bQ$ becomes {\em tailless}. Tailless is the situation in which no node in $G$ has $succ=\bot$.  However, this phenomenon happens in $\bQ$ for just $n$ rounds which follows immediately from Lemma \ref{lemma:onequeueeveryOn}.

\begin{corollary}
\label{corollary:tailless}
The queue formed is tailless for $O(n)$ rounds.
\end{corollary}

We now prove the correctness properties of Algorithm \ref{algorithm:queue1interval} in the sense that it eventually forms a distributed queue so that every \queue\ request is enqueued in $\bQ$ and
each \queue\ request is enqueued only once.

\begin{lemma}
\label{lemma:enqueuedonce1I}
Each {\queue} request is enqueued in $\bQ$ only once.
\end{lemma}
\begin{proof}
We have from Lemma \ref{lemma:onequeueeveryOn} that each queue request is enqueued in $\bQ$ within finite number of rounds after it is issued.
To prove that each \queue\ request is enqueued in $\bQ$ only once,
recall that initially every node $u\in V$ has $succ_u =\infty$ except the head node of $\bQ$ which has $succ_{head} = \bot$. According to Algorithm \ref{algorithm:queue1interval}, no \queue\ request can make itself the successor of any node in $G$ for which $succ_i = \infty$ or $succ_i = j$, where $j$ is the UID of some node in graph $G$ such that $j\neq i$.
%
%
%
%
In Algorithm \ref{algorithm:queue1interval}, we have that each node $u$ changes the value of its local variable $succ_u$ from $\infty$ to $\bot$ only after the {\queue} request from it joined $\bQ$ at node $x$ at the end of the search phase such that $succ_x = u$ (i.e., $u$ becomes the tail of $\bQ$) and the \cancel\ message generated at $x$ (the predecessor node of $u$ in $\bQ$) to remove replicated {\queue} message for the queue request of  $u$ (from other nodes in $G$ except $x$ and $u$) reaches $u$ at the end of the cancelation phase, the current tail of $\bQ$. Therefore, only one {\queue} request can see $succ_l=\bot$ at some node $l$ such that some pending \queue\ request from node $o$ that is currently at node $l$ 
can make $succ_l=o$ at the end of every cycle. After $o$ becomes the successor of $l$, there is no node $p$ in the system with $succ_p=\bot$ until a \cancel\ message from $l$ reaches $o$ and $o$ changes its $succ_o$ variable value from $\infty$ to $\bot$ at the end of the cycle. Arguing similar to Lemma \ref{lemma:onequeueeveryOn}, any change in the $succ$ variable for any node in done after $n$ rounds of message exchanges. The first change is done in the node with $succ=\bot$ at the end of a search phase to make it point to some requesting node $u$ and the second change is done in $u$ at the end of a cancelation phase to make $succ_u=\bot$ from $succ_u=\infty$.
The \queue\ request that is enqueued in $\bQ$ in search phase is removed from the system in cancelation phase so that there is no chance of that request being enqueued in $\bQ$ again in the future. Therefore, every request in enqueued in $\bQ$ and each queue request in enqueued exactly once.
\end{proof}

\paragraph{Complexity in Sequential Executions:}
We prove here the round complexity of Algorithm \ref{algorithm:queue1interval} in forming $\bQ$ for the set $\cE$ of $k$ \queue\ requests from $k$ different nodes of $G$.
We first prove the round complexity of Algorithm \ref{algorithm:queue1interval} in sequential execution of {\queue} requests. A sequential execution consists of a non-overlapping sequence of {\queue} operations. As {\queue} requests do not overlap with each other in sequential executions, the system attains quiescent configuration (no message is in transit and no sequence of events in which a message is sent) after a \queue\ request is served and until a next \queue\ request is issued, i.e. the next \queue\ request will be issued only after the current \queue\ request finishes. We provide the tight bound for Algorithm \ref{algorithm:queue1interval} in sequential executions.

\begin{theorem}
\label{theorem:sequential-phase}
Algorithm \ref{algorithm:queue1interval} is optimal for the distributed queuing problem in sequential executions.
\end{theorem}
\begin{proof}
According to Lemma \ref{lemma:onequeueeveryOn}, $\beta_i$ {\queue} requests in the beginning of cycle $i$ join $\bQ$ (i.e., find their predecessor nodes)  within next $O(n\beta_i)$ rounds starting from the beginning of the cycle $i$. Since $\beta_i=1$ in every cycle $i$ in sequential executions and there are  $k$ {\queue} requests in the system, Algorithm \ref{algorithm:queue1interval} needs $O(nk)$ rounds, in the worst-case.

We now show that this round complexity is the best possible any distributed queuing algorithm can do in sequential executions in 1-interval connected graphs. We prove that, in sequential executions, any algorithm for the distributed queuing problem in 1-interval connected graphs requires at least $\Omega(nk)$ rounds to complete against a strong adversary. We borrow some ideas from \cite{Dutta2013} for this proof.
Consider a set $\cE = \{q_1,q_2,\ldots,q_k\}$ of $k$ {\queue} requests. As {\queue} request do not overlap with each other in sequential executions, we focus our attention on the least number of rounds needed to serve one {\queue} request. The lower bound then follows by amplifying the number of rounds needed for one request to all $k$ requests in $\cE$.
We proceed as follows. Let the node $u$ issued the {\queue} request $q_0$ and node $v$ is the current tail node of the queue with $succ_v=\bot$. To finish execution of $q_0$, $q_0$ should be reached to $v$ and change the existing value of $succ_v$ such that $succ_v=u$. The adversary can connect nodes $u, v_1, \ldots, v_{n-2}, v$ in $G$  in a line in the first round thereby guaranteeing only node $v_1$ gets $q_0$. In the next round, the adversary connects $u, v_2, \ldots, v_{n-2}, v_1$  in a line. In this round, node $v_2$ and $v_{n-2}$ will both get {\queue} message $q_0$. The adversary can continue this way for $\frac{n-2}{2}+1$ rounds, at which point the \queue\ message $q_0$ from node $u$ will eventually reach the tail node $v$ with $succ_v=\bot$. After changing $succ_v$ to $u$ such that $u$ becomes the new tail, the corresponding {\cancel} messages needs also $\frac{n-2}{2}+1$ rounds to reach to node $u$ from $v$. That is, we need  $2(\frac{n-2}{2}+1)= n$ rounds to serve the queuing request $q_0$. Repeating this argument for all the $k$ {\queue} requests in $\cE$, we have the lower bound of $\Omega(nk)$ rounds, as needed.
\end{proof}

\paragraph{Complexity in Concurrent Executions:}
We now consider the round complexity of Algorithm \ref{algorithm:queue1interval} in concurrent one-shot execution of {\queue} requests.
We assume the $R\subseteq V, |R|=k,$ nodes in the graph $G$ issue one {\queue} request each at round 0 and no further {\queue} requests occur.
We prove the following theorem.

\begin{theorem}
\label{theorem:concurrent-phase}
Algorithm \ref{algorithm:queue1interval} solves the distributed queuing problem in $O(nk)$ rounds in concurrent executions.
\end{theorem}
\begin{proof}
According to Algorithm \ref{algorithm:queue1interval}, in the worst-case execution scenario, we can order the \queue\ requests 
in such  way that the smallest \queue\ request (with respect to the lexicographical ordering of active \queue\ requests) ordered first and the largest \queue\ request ordered last. As initiation time is same for all $k$ \queue\ requests in concurrent executions, the ordering only depends on the UID of requesting nodes. 
Therefore, the successor of the $head$ of $\bQ$ is the smallest UID node among the nodes that issued \queue\ requests, the successor of the head's successor node is the second smallest UID node among the nodes that issued \queue\ requests, and so on. The \queue\ request from the highest UID node ordered last in $\bQ$.
Since we consider 1-interval connected graphs and all $k$ request come at the same time in the beginning of execution, we have that $\beta_1=k$ in the beginning of the first cycle. As no more request arrives in the system later in the execution, $\beta_i$ decreases in every cycle $i > 1$. Therefore, using Lemma \ref{lemma:onequeueeveryOn} and replacing $\beta_i$ by $k$, the theorem follows.
\end{proof}

\paragraph{Complexity in Dynamic Executions:}
We now consider the round complexity of Algorithm \ref{algorithm:queue1interval} in dynamic execution of {\queue} requests.
We assume the $R\subseteq V, |R|=k,$ nodes in the graph $G$ issue one {\queue} request each at arbitrary moments of time.
We prove the following theorem.

\begin{theorem}
\label{theorem:dynamic-phase}
Algorithm \ref{algorithm:queue1interval} solves the distributed queuing problem in $O(nk)$ rounds in dynamic executions.
\end{theorem}

\begin{proof}
We proved in Lemma \ref{lemma:onequeueeveryOn} that when a \queue\ request $q$ is issued in the arbitrary round $r$, and there are $\beta_i$ active \queue\ requests in the system which have the initiation times less than $r$, then the request $q$ will be enqueued in $\bQ$ within next $O(n\beta_i)$ rounds starting from the round $r$. 
Therefore, the round complexity of Algorithm \ref{algorithm:queue1interval} is dynamic executions is no more than the round complexity bounds proved in Theorems \ref{theorem:sequential-phase} and \ref{theorem:concurrent-phase}.
\end{proof}

\section{Queuing in More Stable Graphs}
\label{section:stable}
We now study whether the distributed queuing problem can be sped up in more stable graphs. We consider $T$-interval connected graphs of Definition \ref{definition:t-interval} and give an algorithm (see Algorithm \ref{algorithm:queue2Tinterval}) to solve the distributed queuing problem for some $T> 1$. This algorithm is also an extension to the token dissemination algorithm  given in \cite{Kuhn2010} for $T$-interval connected graphs.
%
The main idea behind Algorithm \ref{algorithm:queue2Tinterval} is to serve $\gamma := \min\{\alpha, T\}$ \queue\ requests in $O(n)$ rounds when the graph is $2T$-interval connected. Note that $\alpha := \min\{\beta_1,\beta_2,\ldots\}$, where $\beta_\ell$ is the number of active \queue\ requests in the beginning of cycle $\ell$. 
If $\alpha=1$ in every cycle $\ell$, this constitutes a sequential execution, whereas there is a one-shot concurrent execution in the case when $\alpha\geq T$ in every cycle $\ell$. However due to the properties of $T$-interval connected graphs, Algorithm \ref{algorithm:queue2Tinterval} can  broadcast only $\gamma=T$ \queue\ requests to all the nodes in $G$ in $O(n)$ rounds in these $2T$-interval connected graphs. 
In dynamic executions, $\gamma$ is between $2$ to $T$ in every cycle $\ell$. In summary, $\alpha$  has the impact in the performance of Algorithm \ref{algorithm:queue2Tinterval} in the sense that it determines how many cycles are needed to form a distributed queue for the active \queue\ requests. Therefore, $\alpha$ essentially represents the {\em concurrency level} of \queue\ requests. $\gamma$ does not necessarily be known to Algorithm \ref{algorithm:queue2Tinterval} in the beginning, it can be adapted based on $\beta_\ell$ and $T$ while in execution.

\begin{algorithm}[!t]
{\footnotesize
$S_u \leftarrow \emptyset$; \tcp*[f]{\queue\ requests already broadcasted by node $u$} \\
$A_u \leftarrow \emptyset$; \tcp*[f]{\queue\ requests already received by node $u$}\\

\BlankLine
{\bf For} $\ell = 0, \ldots, \lceil k/\gamma \rceil -1$  {\bf do}  \tcp*[f]{$\gamma := \min\{\alpha,T\}$} \\
{\Indp
{\bf For} $\eta = 0, \ldots, \lceil n/T\rceil -1$  {\bf do}\\
{\Indp
{\bf For} $r = 0, \ldots, 2T -1$  {\bf do}\\
{\Indp
{\bf If} $S_u \neq A_u$ {\bf then}\\
{\Indp
$q_{\min} \leftarrow$ a \queue\ message in  $A_u\backslash S_u$ that is smallest w.r.t. lexicographical ordering on the initiation round and the identifier of the issuing node, respectively;\\
{\bf broadcast} $q_{\min}$ to neighbors;\\
$S_u \leftarrow S_u \bigcup \{q_{\min}\}$;\\
}
{\bf receive} \queue\ messages 
 from $s\geq 1$ neighbors;\\
$A_u \leftarrow A_u \bigcup \{q_1, \cdots, q_s\}$;\\
}
$S_u \leftarrow \emptyset$;\\
}
{\bf If} $succ_u == \bot$ {\bf then}\\
{\Indp
$t \leftarrow$ UID of the smallest \queue\ message in $A_u$ w.r.t. the lexicographical ordering;\\
$succ_u \leftarrow t$;\\
}
{\bf If} a \queue\ request $q \in A_u$ is $jth\_Smallest(A_u)$ for $1 < j < \gamma$ w.r.t. the lexicographical ordering and the UID of $q$ is equal to the $UID$ of node $u$ {\bf then}\\
{\Indp
$t\leftarrow$ UID of a $(j+1)th\_Smallest(A_u)$ \queue\ message w.r.t. the lexicographical ordering;\\
$succ_u \leftarrow t$;\\
}
{\bf If} a \queue\ request $q\in A_u$ is $\gamma th\_Smallest(A_u)$, $\gamma>1$, and the UID of $q$ is equal to the UID of $u$ 
{\bf then} \\
{\Indp
$succ_u \leftarrow \bot$;\\
}
$A_u\leftarrow A_u$ after removing $\gamma$ smallest \queue\ messages from $A_u$;\\ 
}}
\caption{A queuing algorithm run by node $u$}
\label{algorithm:queue2Tinterval}
\end{algorithm}

Algorithm \ref{algorithm:queue2Tinterval} consists of $\lceil k/\gamma\rceil$ cycles. 
In contrast to Algorithm \ref{algorithm:queue1interval}, we do not need cancelation phase in this algorithm as $\gamma$ smallest \queue\ requests can be queued after $O(n)$ rounds and then corresponding \queue\ requests that are replicated to other nodes can be implicitly canceled.
Moreover, each cycle consists of $\lceil n/T\rceil$ periods of $2T$ rounds each, i.e., there are total $2n$ rounds in each cycle (Lines 4, 5 of Algorithm \ref{algorithm:queue2Tinterval}). During each period, each node $u$ maintains the set $A_u$ of \queue\ messages it has already learned and a set $S_u$ of \queue\ messages it has already broadcasted in the current period. $S_u$ is initially empty and it is made empty after each period $\eta$. 

The main idea behind Algorithm \ref{algorithm:queue2Tinterval} is to be able to enqueue $\gamma$ {\queue} requests from $A_u$ in $O(n)$ rounds. We exploit the $T$-interval connectivity and the concurrent level parameter $\gamma$ to perform this task as follows. In each round of the period (Lines 5--11 of Algorithm \ref{algorithm:queue2Tinterval}), each node $u\in V$ selects the smallest \queue\ message $q_{\min}$ that is in $A_u\backslash S_u$ with respect to the lexicographical ordering based on the initiation round and the UID of the \queue\ request issuing node (Line 7 of Algorithm \ref{algorithm:queue2Tinterval}). The node $u$ then broadcasts $q_{\min}$ to its neighbors and adds $q_{\min}$ to the set $S_u$ (Lines 8, 9 of Algorithm \ref{algorithm:queue2Tinterval}).
%
As a stable connected subgraph $G_\eta$ persists for each period, we can always send in a round of the period the token that was not already broadcasted. As $G_\eta$ changes in the next period, wet set $S_u$ (the set of {\queue} requests already broadcasted by node $u$) to $\emptyset$ (Line 12 of Algorithm \ref{algorithm:queue2Tinterval}) and start broadcasting similarly in the next round. This is to make sure that the neighboring nodes of $u$ in the new connected graph $G_\eta'$ receive the tokens that were received by the neighboring nodes in the previous period.
After repeating this process for $\lceil n/T\rceil$ periods, we check the local variable $succ_u$ of each node $u \in G$ to see whether $succ_u$ is $\bot$. If $succ_u==\bot$ for some node $u$, then this must be the tail node of $\bQ$ that was formed in previous cycle, so we select the smallest \queue\ message $q_{\min}$ from $A_u$ 
and
assign the UID $t$ associated with $q_{\min}$ to $succ_u$, i.e. $succ_u \leftarrow t$ (Lines 13--15 of Algorithm \ref{algorithm:queue2Tinterval}). 

To complete the queuing of $\gamma$ \queue\ requests in a cycle, we perform the following before next cycle begins. If some node $u$ issued a \queue\ request $q$ such that $q$ is the $j$th smallest request in $A_u$ for $1<j<\gamma$ and the UID of $q$ is equal to the UID of a node $u\in V$, then we set $succ_u \leftarrow t$, where $t$ is the UID of the $(j+1)$th smallest request in $A_u$ (Lines 16--18 of Algorithm \ref{algorithm:queue2Tinterval}). This is also determined based on the lexicographical ordering on initiation time and UIDs associated with the requests in $A_u$. After that, $succ_u$ is set to $\bot$ for the $\gamma$th smallest \queue\ request issuing node (Lines 19, 20 of Algorithm \ref{algorithm:queue2Tinterval}). At the end of each cycle, we remove all the $\gamma$ requests that joined $\bQ$ so that only remaining requests compete to join $\bQ$ in the next cycle (Line 21 of Algorithm \ref{algorithm:queue2Tinterval}).

\subsection{Analysis}
\label{subsection:analysis-2T}

Similar to Algorithm \ref{algorithm:queue1interval}, we first establish progress and correctness properties of Algorithm \ref{algorithm:queue2Tinterval}. We consider the execution of continuous arrival of \queue\ requests (i.e., $k\rightarrow \infty$) similar to Lemma \ref{lemma:enqueuedonce1I}.

\begin{lemma}
\label{lemma:alpha-enqueue}
If there are $\beta_\ell$ active \queue\ requests in the beginning of any cycle $\ell$ in a dynamic execution, Algorithm \ref{algorithm:queue2Tinterval} guarantees that they will be enqueued in $\bQ$ in next at most $O\left(n+\frac{n\beta_\ell}{\min\{\beta_\ell,T\}}\right)$ rounds.
\end{lemma}

\begin{proof}
Recall that a \queue\ request that is initiated in the beginning of a cycle can join $\bQ$ after it reaches a node $x$ such that $succ_x=\bot$ at the end of a cycle, assuming that there is no other \queue\ request in the system. 
We know from $2T$-interval connectivity of the graph that there is a stable connected subgraph $G_\eta$ in each period $\eta$ that does not change throughout the period of $2T$ rounds.
Therefore, through the pipelined broadcasting of the \queue\ requests in each round, if there are $\beta_\ell$ active \queue\ requests in the beginning of a cycle $\ell$, we prove here that $\min\{\beta_\ell,T\}$ \queue\ requests will reach to all the nodes in $G$ at the end of the cycle $\ell$. Therefore, if $\beta_\ell \leq T$, all the requests reach to all the nodes in $G$ at the end of that cycle, but in the case when $\beta_\ell >T$ then we need $\beta_\ell/\min\{\beta_\ell,T\}$ cycles to finish all the $\beta_\ell$ requests.

We proceed as follows similar to \cite{Kuhn2010} for each cycle $\ell$ of Algorithm \ref{algorithm:queue2Tinterval}. Let $K_\eta(q)$ denote the set of nodes that know a \queue\ request $q$ at the beginning of period $\eta$ and let $\dist_\eta(u,q)$ denote the minimum distance in $G_\eta$ between a node $u$ and any node that is in $K_\eta(q)$. Let $A_u^\eta(r)$ and  $S_u^\eta(r)$ denote the values of the local sets $A_u$ and $S_u$ of node $u$ at the beginning of round $r$ of period $\eta$. Note that the node $u$ knows a \queue\ message $q$ whenever $q\in A_u$. According to the definition of $2T$-interval connectivity,  if a round $r$ is such that $\dist_\eta(u,q)\leq r\leq 2T$, then either $q$ belongs to  $S_u^\eta(r+1)$ or $S_u(r+1)$ includes at least $r-\dist_\eta(u,q)$ {\queue} requests that are smaller than $q$ with respect to the lexicographical ordering of \queue\ requests. 
Therefore, if $r\geq \dist_\eta(u,q)$, then $r$ rounds must be enough for the node $u$ to receive the {\queue} request $q$. Moreover, if  $r\geq \dist_\eta(u,q)$ but $u$ has not received $q$, then there must be smaller {\queue} requests than $q$ from other nodes that have blocked the broadcast of request $q$ in nodes that are between $u$ and the node that initiated $q$.

Now we show that at the end of each cycle $\ell$, at least $\min\{\beta_\ell,T\}$ smallest \queue\ requests among the $\beta_\ell$ active \queue\ requests that are available in the system in the beginning of cycle $\ell$ are reached to all the nodes and then they can be enqueued in $\bQ$.
Again, we proceed similar to \cite{Kuhn2010}. Let $N_\eta^d(q):=\{u\in V|\dist_\eta(u,q)\leq d\}$ denote the set of nodes at distance at most $d$ from some node that knows $q$ at the beginning of period $\eta$ and let $q$ be one of the $\min\{\beta_\ell,T\}$ smallest {\queue} request with respect to the lexicographical ordering of \queue\ requests. We have that, for each node $u\in N_\eta^T(q)$, either $q \in  S_u^\eta(2T+1)$ or $S_u^\eta(2T+1)$ contains at least $\min\{\beta_\ell,T\}$ {\queue} requests which are smaller than $q$. As $q$ is one of the smallest {\queue} request, this is not the case that $S_u^\eta(2T+1)$ contains at least $\min\{\beta_\ell,T\}$ {\queue} requests which are smaller than $q$. Therefore, all nodes in $N_\eta^T(q)$ know \queue\ request $q$ at the end of the period $\eta$. As $G_\eta$ is connected, at each period $T$ new nodes learn $q$. 
Since there are no more than $n$ nodes in the network $G$ and we have $\lceil n/T\rceil$ periods, at the end of the last period, all nodes know $q$. Therefore, at least $\min\{\beta_\ell,T\}$ smallest \queue\ request will be at all nodes in $G$ at the end of each cycle $\ell$. These $\min\{\beta_\ell,T\}$ smallest \queue\ requests are then implicity enqueued in $\bQ$ before the next cycle $\ell+1$ begins (Lines 13--20 of Algorithm \ref{algorithm:queue2Tinterval}).
We have that each cycle $\ell$ consists of $\lceil n/T \rceil$ periods of $2T$ rounds each. That is, we have $2n$ rounds in a cycle.
Moreover, as we use initiation time in finding the $\min\{\beta_\ell,T\}$ smallest \queue\ requests, no quest request that is initiated during cycle $\ell$ or later overtakes the requests $\bB_i$ that are initiated up to the beginning of cycle $\ell$. Therefore, all the $\beta_\ell$ requests will be enqueued in $\bQ$ in next
at most $O\left(n+ \frac{n \beta_i}{\min\{\beta_\ell,T\}}\right)$ rounds.
\end{proof}

\begin{lemma}
\label{lemma:atmostonce-2T}
Algorithm \ref{algorithm:queue2Tinterval} enqueues each \queue\ request in $\bQ$ only once.
\end{lemma}

\begin{proof}
We prove this lemma similar to Lemma \ref{lemma:enqueuedonce1I}. Recall the every node $u$ in the system initially has $succ_u=\infty$ except the head node of $\bQ$ which has $succ_{head}= \bot$. In Algorithm \ref{algorithm:queue2Tinterval}, the enqueue of $\min\{\beta_\ell,T\}$ \queue\ requests to $\bQ$ happens at the end of each cycle (Lines 13--20 of Algorithm \ref{algorithm:queue2Tinterval}). In this process, the node $u$ which has $succ_u=\bot$ changes its value from $\bot$ to $t$, where $t$ is the UID of the smallest \queue\ message in $A_u$ with respect to the lexicographical ordering of the \queue\ requests in $A_u$. After that the second smallest to $\min\{\beta_\ell,T\} - 1$ smallest \queue\ request are enqueued implicitly as given in Lines 16--18 of Algorithm \ref{algorithm:queue2Tinterval}. The local successor variable $succ_u$ of the node $u$ that issued the $\min\{\beta_\ell,T\}$th smallest \queue\ message is set to $\bot$. As all $\min\{\beta_\ell,T\}$ smallest \queue\ requests are removed from $A_u$ at the end of each cycle $\ell$, after this enqueue they can not be enqueued in the future. Therefore, in this process, each \queue\ request is enqueued in $\bQ$ only once. Moreover, Algorithm \ref{algorithm:queue2Tinterval} does not terminate until all requests in $\cE$ finished execution. Hence, the lemma follows.
\end{proof}

We now analyze the performance of Algorithm \ref{algorithm:queue2Tinterval} in sequential, concurrent, and dynamic executions.

\paragraph{Complexity in Sequential Executions:}
We show that, for the sequential execution of $k$ {\queue} requests, the distributed queuing problem needs $\Theta(nk)$ rounds to solve in the worst-case even in $T$-interval connected graphs.

\begin{theorem}
\label{theorem:sequential-2T}
In sequential executions, Algorithm \ref{algorithm:queue2Tinterval} is optimal for the distributed queuing problem in $T$-interval connected graphs against a strong adversary.
\end{theorem}

\begin{proof}
Recall that \queue\ requests do not overlap with each other in sequential executions. The upper bound of $O(nk)$ is immediate from Theorem \ref{theorem:sequential-phase} as each \queue\ request is enqueued in $\bQ$ at the end of each cycle in the worst-case, irrespective of the $T$-interval connectivity. We now focus our attention to prove the lower bound of $\Omega(nk)$ in $T$-interval connected graphs. The idea of the proof is also similar the lower bound proof of Theorem \ref{theorem:sequential-phase}. As there is only one \queue\ request $q$ in the system at any time in sequential executions, the adversary can connect the nodes in a line for $T$ rounds in such a way that only one new node can learn $q$ in each round. The adversary can repeat this again for next $T$ rounds by connecting the nodes of the graph in a line, so that only other $T$ nodes can learn $q$. Therefore, $q$ needs $n$ rounds (i.e. a cycle) to reach to the tail of $\bQ$ and join it to become a new tail of $\bQ$. Repeating this argument for all the $k$ \queue\ requests in $\cE$, the lower bound follows, as needed.
\end{proof}

\paragraph{Complexity in Concurrent Executions:}
We prove the following theorem
for the performance of Algorithm \ref{algorithm:queue2Tinterval} on the concurrent (one-shot) execution of $k$ {\queue} requests. 

\begin{theorem}
\label{theorem:concurrent-2T}
In concurrent executions, Algorithm \ref{algorithm:queue2Tinterval} requires $O(n+\frac{nk}{T})$ rounds to solve the distributed queuing problem in $T$-interval connect graphs.
\end{theorem}

\begin{proof}
Since all the \queue\ requests in $\cE$ arrive in the system in the beginning of the first cycle, we have from Lemma \ref{lemma:alpha-enqueue} that $T$ \queue\ requests will be enqueued in $\bQ$ at the end of the first cycle.
As this needs to repeat up to $\lceil k/T\rceil$ times to make sure that all the $k$ requests joined $\bQ$, we need $O(n+\frac{nk}{T})$ rounds to to serve all $k$ \queue\ requests in $\cE$.
\end{proof}

\paragraph{Complexity in Dynamic Executions:} We prove the following theorem for the performance of Algorithm \ref{algorithm:queue2Tinterval} in dynamic execution of $k$ queue requests.

\begin{theorem}
\label{theorem:dynamic-2T}
In dynamic executions, Algorithm \ref{algorithm:queue2Tinterval} requires $O\left(n+\frac{nk}{\min\{\alpha, T\}}\right)$ rounds to solve the distributed queuing problem in $T$-interval connected graphs.
\end{theorem}

\begin{proof}
In a cycle $\ell$, Algorithm \ref{algorithm:queue2Tinterval} can enqueue  $\min\{\beta_\ell,T\}$ \queue\ requests that are initiated in the cycles up to the beginning of cycle $\ell$.  It can be seen from Theorem \ref{theorem:concurrent-2T} that if $\beta_\ell > T$ then the round complexity of Algorithm \ref{algorithm:queue2Tinterval} depends on the value of $T$. If $\beta_\ell < T$, Algorithm \ref{algorithm:queue2Tinterval} can not exploit the benefits of $T$-interval connectivity and only $\beta_\ell$ \queue\ requests can be enqueued in $\bQ$ at the end of each cycle. Therefore, as only $\min\{\beta_\ell,T\}$ requests can be enqueued in each cycle $\ell$ based on the concurrency level parameter $\beta_\ell$ in each cycle $\ell$, arguing similar to Theorem \ref{theorem:concurrent-2T}, we need to run Algorithm \ref{algorithm:queue2Tinterval} for at most $\lceil \frac{ k}{\min\{\alpha,T\}}\rceil$ cycles to make sure that all $k$ \queue\ requests joined $\bQ$, where $\alpha := \min\{\beta_1,\beta_2,\ldots\}$ for the value of $\beta_\ell$ in each cycle $\ell$. Thus, Algorithm \ref{algorithm:queue2Tinterval} needs  $O\left(n+\frac{nk}{\min\{\alpha, T\}}\right)$ rounds to serve all $k$ \queue\ requests in a dynamic execution of $k$ requests. 
\end{proof}

Theorem \ref{theorem:dynamic-2T} subsumes the results in Theorems \ref{theorem:sequential-2T} and \ref{theorem:concurrent-2T} in the sense that the round complexity bound
of Theorem \ref{theorem:dynamic-2T} becomes $O(nk)$ as $\min\{\alpha,T\}=1$ in every round of any sequential execution and becomes $O(n+\frac{nk}{T})$ as $\min\{\alpha,T\}=T$ in every round of any concurrent execution.

We assumed in Algorithm \ref{algorithm:queue2Tinterval} that $T$ is known. If $T$ is not known then we can guess $T$ by trying all the values of $T = 1 ,2, 4, \cdots, k$. This incurs extra $\log k$ factor in the round complexity bound. Therefore, we can solve the distributed queuing problem in $O\left(\min\left\{nk, n\log k + \frac{nk\cdot \log k}{\min\{\alpha,T\}}\right\}\right)$ rounds in any execution.

\section{An Inherent Limitation}
\label{section:limitation}

We discuss here why algorithms designed for the distributed queuing problem in the adversarial dynamic graph model, including Algorithms \ref{algorithm:queue1interval} and \ref{algorithm:queue2Tinterval}, need to perform $n$ consecutive rounds of message broadcasts before they enqueue some \queue\ requests in the distributed queue $\bQ$. In other words, we argue why we used explicit cycles of $n$ consecutive rounds for message broadcasts in our algorithms before we decide to enqueue any \queue\ request in $\bQ$.  Our argument is under the assumption that the queue $\bQ$ formed by any queuing algorithm needs to
ensure the following two properties which together provide the {\em correctness} of the distributed queue formed.

\begin{enumerate}
\item  Each \queue\ request is eventually enqueued in $\bQ$ after it is issued.  This guarantees that
    no \queue\ request is canceled (or removed) from the system without being enqueued in $\bQ$, after it is issued.

\item Each \queue\ request is enqueued in $\bQ$ exactly once. This property guarantees that no \queue\ request is enqueued in $\bQ$ more than one time.
\end{enumerate}

These two properties imply that every request will be enqueued in $\bQ$ but only once. Our objective now is to present some instances of the distributed queuing problem where it is difficult to satisfy these two properties simultaneously if we allow any algorithm for this problem to enqueue some queue requests in $\bQ$ within $o(n)$ rounds of message broadcasts after the last enqueue by
that algorithm. 
In particular, we present two instances of the distributed queuing problem.
We consider the dynamic execution 
in 1-interval connected graphs in this discussion; recall that \queue\ requests are initiated in arbitrary moments of time in a dynamic execution.

We start with the first instance where we try to satisfy the second property 
from which the first property is violated.
Let the queuing algorithm that we consider in this discussion allows the tail node $p$ in $\bQ$ enqueue a \queue\ request $q$ from any node $v$ as soon as it receives $q$.
Consider an execution instance in which some node $u$ that issued a \queue\ request $q$ in some round $i-t, t \leq o(n),$ reached the current tail node $p$ (with the local successor variable $succ_p=\bot$) at round $i$ such that $p$ can now made $u$ its successor (the new tail of $\bQ$), 
that is $succ_p=u$. 
In $t$ consecutive rounds of message broadcasting $q$ might also have been replicated to some other nodes in the network because \queue\ message replication is necessary (Theorem \ref{theorem:impossibility}) to solve the queuing problem.
As $q$ is already enqueued in $\bQ$, to satisfy the second property so that it will not be enqueued in $\bQ$ more than once, $q$ has to be removed from those nodes so that it will not be enqueued again in $\bQ$. 
As nodes have no global information, the nodes where $q$ still exists
need to rely on removing either the largest or the smallest \queue\ message using some ordering mechanism (e.g., UIDs of \queue\ request issuing nodes, initiation times, or the combination of both) from the set of requests that are at those nodes at round $i$. Lets assume that, at round $i$, in two nodes $u'$ and $u''$ of the graph $G(i)$, $q' (\neq q)$ is the smallest \queue\ request as $q$ has not yet been reached to $u'$ and $u''$, and $q$ is the smallest \queue\ request in all the remaining nodes of the graph. 
Now when a queuing algorithm uses the technique to remove the smallest \queue\ request from all the network nodes, $q'$ will be removed from $u'$ and $u''$ which was not yet enqueued in $\bQ$ and $q$ will be removed from rest of the nodes in the graph, so that there is no possibility that $q$ will be enqueued twice in $\bQ$, satisfying the second property. But, this violates the first property because some other request was removed from the system before it has been enqueued in $\bQ$.
However, if the algorithm would have allowed $t=n$ rounds of message broadcasts before it enqueue $q$, $q$ would have been the smallest request in all the nodes in the graph and both properties would have been satisfied.
As the graph is controlled by a strong adversary, sending the acknowledgement messages to remove the particular requests from the nodes also need $\frac{n-2}{2}+1$ rounds in the worst-case as adversary can give very bad graph in every round (Theorem \ref{theorem:sequential-phase}), forcing the acknowledgement to reach one of the required nodes after $\frac{n-2}{2}+1$ rounds.

We discuss now the second execution instance where we try to satisfy the first property from which the second property is violated. Consider the above mentioned execution instance and assume that $p$ does not try to remove $q$ immediately. Instead $p$ tries to send acknowledgement (cancel) messages to nodes where $q$ has been replicated. Suppose an acknowledgement message is reached to $u$ at round $i+s$, where $s\leq o(n)$, and some other queue request $q''$ from node $w$ that was at $u$ became the new tail of $\bQ$. Now $u$ issues an acknowledgement message for $q''$. As $s$ is very small, the acknowledge message for $q$ (from $p$) may not have been reached already to all the nodes where $q$ still exists. Let $w$ be the node where $q$ is the only request that it is has. Let, at round $i+s+1$, acknowledgement message from $u$ reached $w$ ($w$ and $u$ happened to be the neighbors in the graph $G(i+s+1)$ given by the adversary); which in turn forces $w$ to make $u$ its successor. This violates the second property as $q$ is enqueued twice in $\bQ$. We summarize our discussion in the following observation which shows that there are some execution instances of the distributed queuing problem where messages broadcast for at least $\frac{n-2}{2}+1$ consecutive rounds is needed for any algorithm before enqueuing any \queue\ request in $\bQ$, in the worst-case.

\begin{observation}
There are execution instances of the distributed queuing problem for which $\Theta(n)$ consecutive rounds of message broadcasts by the graph nodes is necessary and sufficient for any algorithm before it enqueues any \queue\ request(s) in a distributed queue $\bQ$ so that $\bQ$ that is formed from the execution of the \queue\ requests in the system is {\em correct} -- each \queue\ request is eventually enqueued in $\bQ$ and no \queue\ request is enqueued in $\bQ$ more than once.  
\end{observation}

\section{Discussion}
\label{section:discussion}

We addressed the distributed queuing problem in adversarial dynamic networks by giving two simple algorithms, one for 1-interval connected graphs and the other for $T$-interval connected graphs. These algorithms work in sequential, concurrent, and dynamic execution instances of the problem.
Our solutions for 1-interval connected graphs can be easily extended to solve this problem in $O(\frac{nk}{c})$ rounds in {\em $c$-vertex connected graphs} for some $c>1$ $-$  we say that a dynamic network $G=(V,E)$ is always $c$-vertex connected if and only if $G(r)$ is $c$-vertex connected for every round $r$, i.e. each node is connected to every $c$ other nodes \cite{Haeupler2012}.
Our results and the discussion in Section \ref{section:limitation} suggest that, in the worst-case, algorithms for the distributed queuing problem need the same number of rounds required for the $k$-token dissemination problem.
Therefore, it is interesting to establish a lower bound similar to the $k$-token dissemination problem given in \cite{Kuhn2010,Dutta2013,Haeupler2012} for the distributed queuing problem in this model; finding faster queuing algorithms is another open problem.
Moreover, Busch and Tirthapura \cite{Busch2010} showed that the related problem of {\em distributed counting}\footnote{In the distributed counting problem,  processors in a distributed system increment a globally unique shared counter. Each processor in return receives the value of a counter after its increment operation took effect \cite{Wattenhofer1998,Aspnes1994}.} is harder than the distributed queuing problem in concurrent situations in static networks. Therefore, it will be very interesting to prove the similar results of \cite{Busch2010} for the distributed queuing and counting problems in this adversarial dynamic network model.

\bibliographystyle{eptcs}
\bibliography{queue}
\end{document}